\documentclass[sigconf]{acmart}

\def\BibTeX{{\rm B\kern-.05em{\sc i\kern-.025em b}\kern-.08emT\kern-.1667em\lower.7ex\hbox{E}\kern-.125emX}}

\usepackage{booktabs} 
\usepackage[ruled,vlined,linesnumbered]{algorithm2e}
\usepackage{hyperref}

\usepackage{balance}

\newcommand{\dis}{{\sc Dispersion}}

\newcommand{{\rh}}{{\widehat r}}
\newcommand{{\Rh}}{{\widehat R}}

\newcommand{\cR}{{\mathcal R}}

\usepackage{mathtools}

\newboolean{short}
\setboolean{short}{false}
\newcommand{\shortOnly}[1]{\ifthenelse{\boolean{short}}{#1}{}}
\newcommand{\onlyShort}[1]{\ifthenelse{\boolean{short}}{#1}{}}
\newcommand{\longOnly}[1]{\ifthenelse{\boolean{short}}{}{#1}}
\newcommand{\onlyLong}[1]{\ifthenelse{\boolean{short}}{}{#1}}
\newcommand{\shortLong}[2]{\ifthenelse{\boolean{short}}{#2}{#1}}
\newcommand{\longShort}[2]{\ifthenelse{\boolean{short}}{#2}{#1}} 
\onlyShort{
\let\labelindent\relax
\usepackage{enumitem}
\setitemize{noitemsep,topsep=0pt,parsep=0pt,partopsep=0pt}
\setenumerate{noitemsep,topsep=0pt,parsep=0pt,partopsep=0pt}
}
    
\renewcommand\footnotetextcopyrightpermission[1]{}
\begin{document}

%

\title{Dispersion of Mobile Robots in the Global Communication Model}

%
\author{Ajay D. Kshemkalyani}
\affiliation{%
  \institution{University of Illinois at Chicago}
  \state{Illinois 60607}
  \country{USA}}
  \email{ajay@uic.edu}

\author{Anisur Rahaman Molla}
\affiliation{%
  \institution{Indian Statistical Institute}
  \city{Kolkata 700108}
  \country{India}}
\email{molla@isical.ac.in}

\author{Gokarna Sharma}
\affiliation{%
  \institution{Kent State University}
  \city{Ohio 44240}
  \country{USA}}
\email{sharma@cs.kent.edu}

%
\renewcommand{\shortauthors}{A. D. Kshemkalyani, A. R. Molla, and G. Sharma}

\begin{abstract}
The dispersion problem on graphs asks $k\leq n$ robots placed initially arbitrarily on the nodes of an $n$-node anonymous graph to reposition autonomously to reach a configuration in which each robot is on a distinct node of the graph. This problem is of significant interest due to its relationship to other fundamental robot coordination problems, such as exploration, scattering, load balancing, and relocation of self-driven electric cars (robots) to recharge stations (nodes). 
In this paper, we consider dispersion in 
the {\em global communication} model where a robot can communicate with any other robot in the graph (but the graph is unknown to robots). 
We provide three novel deterministic algorithms, two for arbitrary graphs and one for arbitrary trees, in a synchronous setting where all robots perform their actions
in every time step. 
For arbitrary graphs, our first algorithm is based on a DFS traversal and guarantees $O(\min(m,k\Delta))$ steps runtime using $\Theta(\log (\max(k,\Delta)))$ bits at each robot, where $m$ is the number of edges and $\Delta$ is the maximum degree of the graph.
The second algorithm for arbitrary graphs is based on a BFS traversal and guarantees $O( \max(D,k) \Delta (D+\Delta))$ steps runtime using $O(\max(D,\Delta \log k))$ bits at each robot, where $D$ is the diameter of the graph. The algorithm for arbitrary trees is also based on a BFS travesal and guarantees $O(D\max(D,k))$ steps runtime using $O(\max(D,\Delta \log k))$ bits at each robot. 
Our results are significant improvements compared to the existing results established in the {\em local communication} model where a robot can communication only with other robots present at the same node.  
Particularly, the DFS-based  algorithm is optimal for both memory and time in constant-degree arbitrary graphs. 
The BFS-based algorithm for arbitrary trees is optimal with respect to runtime when $k\leq O(D)$.
\end{abstract}

%

\maketitle
\sloppy
\section{Introduction}
The dispersion of autonomous mobile robots to spread them out evenly
in a region is a problem of significant interest in distributed robotics, e.g., see \cite{Hsiang:2003,HsiangABFM02}. Recently, this problem has been formulated by Augustine and Moses Jr.~\cite{Augustine:2018} in the context of graphs. 
They defined  the problem as follows:  Given any arbitrary initial configuration of $k\leq n$ robots positioned on the nodes of an $n$-node anonymous graph, the robots reposition autonomously to reach a configuration where each robot is positioned on a distinct node of the graph (which we call the {\dis} problem).  
This problem has many practical applications, for example, in relocating self-driven electric cars (robots) to recharge stations (nodes), assuming that the cars have smart devices to communicate with each other to find a free/empty charging station \cite{Augustine:2018,Kshemkalyani}.
This problem is also important due to its relationship to many other well-studied autonomous robot coordination problems,
such as  exploration, scattering, load balancing, covering, and self-deployment~\cite{Augustine:2018,Kshemkalyani}. One of the key aspects of mobile-robot research is to understand how to use the resource-limited robots to accomplish some large task in a distributed manner \cite{Flocchini2012,flocchini2019}. 

In this paper, we continue our study on the trade-off between memory requirement and time to solve {\dis} on graphs.
We consider for the very first time the problem of dispersion in the {\em global communication} model where a robot can communicate with any other robot in the system (but the graph structure is not known to robots). The previous work  \cite{Augustine:2018,Kshemkalyani,Kshemkalyani2019} (details in Tables \ref{table:comparision}, \ref{table:comparision-tree}, and related work) on {\dis} considered the {\em local communication} model where a robot can only communicate with other robots that are present at the same node.
Although  the global communication model seems stronger than the local model  in the first sight, many challenges that occur in the local model also arise in the global model. 
For example, two robots in two neighboring nodes of $G$ cannot figure out just by communication which edge of the nodes leads to each other.  Therefore, the robots still need to explore through the edges as in the local model.  
The global communication model
has been considered heavily in the past in distributed robotics, e.g., see \cite{Das18,FraigniaudGKP06,Ortolf:2012}, in addition to the local model, and our goal is to explore how much global communication helps for {\dis} in graphs compared to the local model.


In this paper, we  
provide three new deterministic algorithms for {\dis} in the global communication model, two for arbitrary graphs and one for arbitrary trees.
Our first algorithm for arbitrary graphs using a {\em depth first search} (DFS) traversal improves by a $O(\log n)$ factor on the state-of-the-art in the local communication model; 
see Table \ref{table:comparision}. 
The second algorithm for arbitrary graphs and the algorithm for arbitrary trees are the first algorithms designed for {\dis} using a {\em breadth first search} (BFS) traversal and provide different time-memory trade-offs.
We also complement our algorithms by some lower bounds on time and memory requirement in the global model. 


\vspace{1mm}
\noindent{\bf Overview of the Model and Results.} We consider the same model (with a only difference as described in the next paragraph) as in Augustine and Moses Jr.~\cite{Augustine:2018}, Kshemkalyani and Ali \cite{Kshemkalyani}, and Kshemkalyani {\it et al.} \cite{Kshemkalyani2019} where a system of $k\leq n$ robots are operating on an $n$-node graph $G$.  
$G$ is assumed to be a connected, undirected graph with $m$ edges, diameter $D$, and maximum degree $\Delta$. In addition, $G$ is {\em anonymous}, i.e.,
nodes have no unique IDs  and hence are indistinguishable but the ports (leading to incident
edges) at each node have unique labels from $[1,\delta]$, where $\delta$ is the degree of that node. 
The robots are {\em distinguishable}, i.e., they have unique IDs in the range $[1,k]$. 
The robot activation setting is {\em synchronous} -- all robots are activated in a round and they perform their operations simultaneously in synchronized rounds. 
Runtime is measured in rounds (or steps).

\begin{table}[H]
{\footnotesize
\centering
\begin{tabular}{cccccc}
\toprule
{\bf Algorithm} & {\bf Memory/robot}   & {\bf Time} & {\bf Comm.~Model/}   \\
 & {\bf (in bits)}   & {\bf (in rounds)} & {\bf Initial Conf.} \\
\toprule
Lower bound & $\Omega(\log(\max(k,\Delta)))$    & $\Omega(k)$ & local  \\
\hline
\cite{Augustine:2018}\footnotemark & $O(\log n)$     & $O(mn)$  & local/general \\
\hline
\cite{Kshemkalyani} & $O(k\log \Delta)$     & $O(m)$  & local/general\\
\hline
\cite{Kshemkalyani} & $O(D \log \Delta)$     & $O(\Delta^D)$  &local/general \\
\hline
\cite{Kshemkalyani} & $O(\log(\max(k,\Delta)))$     & $O(mk)$ & local/general\\
\hline
\cite{Kshemkalyani2019} & $O(\log n)$      & $O(\min(m,k\Delta)\cdot \log k)$  & local/general\\
\hline
\hline
{\bf Lower bound} & $\Omega(\log(\max(k,\Delta)))$    & $\Omega(k)$ & global  \\
\hline
{\bf Thm.~\ref{theorem:0}} & $O(\log(\max(k,\Delta)))$      & $O(\min(m,k\Delta))$  & global/general\\
\hline
{\bf Thm.~\ref{theorem:1}(a)} & $O(\max(D,\Delta \log k))$      & $O(D\Delta(D+\Delta))$  & local/rooted\\
\hline
{\bf Thm.~\ref{theorem:1}(b)} & $O(\max(D,\Delta \log k))$      & $O(\max(D,k)\Delta(D+\Delta))$  & global/general\\
\bottomrule
\end{tabular}
\caption{The results on {\dis} for $k\leq n$ robots on $n$-node arbitrary graphs with $m$ edges, $D$ diameter, and $\Delta$ maximum degree. $^1$The results in \cite{Augustine:2018} are only for $k=n$. 
}  
\label{table:comparision}
}
\vspace{-8mm}
\end{table}

\begin{table}[H]
{\footnotesize
\centering
\begin{tabular}{cccccc}
\toprule
{\bf Algorithm} & {\bf Memory/robot}   & {\bf Time} & {\bf Comm.~Model/}   \\
 & {\bf (in bits)}   & {\bf (in rounds)} & {\bf Initial Conf.} \\
\toprule
Lower bound & $\Omega(\log(\max(k,\Delta)))$    & $\Omega(D^2)$ & local  \\
\hline
\cite{Augustine:2018} & $O(\log n)$     & $O(n)$ & local/general \\
\hline
\hline
{\bf Lower bound} & $\Omega(\log(\max(k,\Delta)))$    & $\Omega(D^2)$ & global  \\
\hline
{\bf Thm.~\ref{theorem:2}(a)} & $O(\max(D,\Delta \log k))$      & $O(D^2)$ & local/rooted \\
\hline
{\bf Thm.~\ref{theorem:2}(b)} & $O(\max(D,\Delta \log k))$      & $O(D\max(D,k))$ & global/general \\
\bottomrule
\end{tabular}
\caption{The results on {\dis} for $k\leq n$ robots on $n$-node trees.  
\label{table:comparision-tree}
}}
\vspace{-6mm}
\end{table}

The only difference with the model in \cite{Augustine:2018,Kshemkalyani,Kshemkalyani2019} is they assume the local communication model -- the robots in the system can communicate with each other only when they are at the same node of $G$, whereas we consider in this paper the global communication model -- the robots in the system can communicate with each other irrespective of their positions. Despite this capability, robots are still oblivious to $G$ and they will not know the positions of the robots that they are communicating to.  

We establish the following two results for {\dis} in an arbitrary graph. 
The second result differentiates the initial configurations of $k\leq n$ robots on $G$. We call the configuration {\em rooted} if all $k\leq n$ robots are on a single node of $G$ in the initial configuration. We call the initial configuration {\em general}, otherwise.

\begin{theorem}
\label{theorem:0}
Given any initial configuration of $k\leq n$ mobile robots in an arbitrary, anonymous 
$n$-node graph $G$ having 
$m$ edges and maximum degree $\Delta$, {\dis} can be solved in $O(\min(m,k\Delta))$ time with $O(\log (\max(k,\Delta)))$ bits  at each robot in the global communication model. 
\end{theorem}

\begin{theorem}
\label{theorem:1}
Given $k\leq n$ mobile robots in an arbitrary, anonymous 
$n$-node graph $G$ having 
$m$ edges, diameter $D$, and maximum degree $\Delta$:
\begin{itemize}
    \item [a.] For the rooted initial configurations,
{\dis} can be solved in $O(D\Delta(D+\Delta))$ time with $O(\max(D,\Delta \log k))$ bits  at each robot in the local communication model. 
\item [b.] For the general initial configurations,
{\dis} can be solved in $O(\max(D,k)\Delta(D+\Delta))$ time with $O(\max(D,\Delta \log k))$ bits  at each robot
in the global communication model. 
\end{itemize}
\end{theorem}

Theorem \ref{theorem:0}  improves by a factor of $O(\log k)$ over the $O(\min(m,k\Delta)\cdot \log k)$ time best previously known algorithm \cite{Kshemkalyani2019} in the local communication model 
(see Table \ref{table:comparision}). 
We also prove a time lower bound of $\Omega(k)$ and memory lower bound of $\Omega(\log(\max(k,\Delta)))$ bits at each robot for {\dis} on graphs in the global communication model. 
The implication is that, for constant-degree arbitrary graphs (i.e., when $\Delta=O(1)$), Theorem \ref{theorem:0} is optimal with respect to both memory and time, first such result for arbitrary graphs.  
%
Theorem \ref{theorem:1} improves significantly on the $O(\Delta^D)$ algorithm in the local communication model. 

We establish the following theorem 
in an arbitrary tree.

\begin{theorem}
\label{theorem:2}
Given $k\leq n$ mobile robots in an anonymous 
$n$-node arbitrary tree  $G$ with degree $\Delta$ and diameter $D$:
\begin{itemize}
    \item [a.] For the rooted initial configurations,
{\dis} can be solved in $O(D^2)$ time with $O(\max(D,\Delta \log k))$ bits  at each robot in the local communication model. 
\item [b.] For the general initial configurations,
{\dis} can be solved in $O(D\max(D,k))$ time with $O(\max(D,\Delta \log k))$ bits  at each robot
in the global communication model.
\end{itemize}
\end{theorem}

We also prove a time lower bound of $\Omega(D^2)$ for {\dis} on trees in the global communication model. The implication is that
the time in Theorem \ref{theorem:2}(a) is optimal. 
Also, the time in Theorem \ref{theorem:1}(a) is optimal for constant-degree arbitrary graphs. 

\vspace{1mm}
\noindent{\bf Challenges and Techniques.}
The well-known  DFS traversal approach \cite{Cormen:2009} was used in the previous results on  {\dis} \cite{Augustine:2018,Kshemkalyani,Kshemkalyani2019}. 
If all $k$ robots are positioned initially on a single node of $G$, then the DFS traversal finishes in $\min(4m-2n+2,\, k\Delta)$ rounds solving {\dis}.  If $k$ robots are initially on $k$ different nodes of $G$, then {\dis} is solved in a single round. 
However, if not all of them are on a single node initially, then the robots on nodes with multiple robots need to reposition (except one) to reach to free nodes and settle.
  The natural approach is to run DFS traversals in parallel to minimize time.

The challenge arises when two or more DFS traversals meet before all robots settle. When this happens, the robots that have not settled yet need to find free nodes. For this, they may need to re-traverse the already traversed part of the graph by the DFS traversal. Kshemkalyani {\it et al.} \cite{Kshemkalyani2019} designed a smarter way to synchronize the parallel DFS traversals so that the total time increases only by a factor of $\log k$ to  $\min(4m-2n+2,k\Delta)\cdot \log k$ rounds, in the worst-case, in the local communication model. However, removing the $O(\log k)$ factor seemed difficult due to the means of synchronization.
We develop in this paper an approach that allows to synchronize DFS traversals without re-traversing the already traversed part of the graph giving us $\min(4m-2n+2,\, k\Delta)$ rounds, as if running DFS starting from all robots in the same node, in the global communication model. 
This is possible due to the information that can be passed to the robots to take their next actions, even if they do not known their positions on $G$. This was not possible in the local communication model and hence the synchronization incurred an $O(\log k)$ factor to synchronize $O(k)$ trees that might be formed in the dispersion process.
For constant-degree arbitrary graphs, the time bound becomes $O(k)$, which is time-optimal. 

Despite efficiency in merging the DFS traversal trees due to global communication, the time bound of $\min(4m-2n+2,\, k\Delta)$ seems to be inherent in algorithms based on a DFS traversal, even in the global model (consider, for example, a rooted initial configuration). The natural way to circumvent this limitation is to run BFS 
traversal to  reach many nodes at once. A naive approach of running BFS gives exponential $O(\Delta^D)$ runtime. Here we design a smarter way of performing BFS so that we can achieve dispersion in arbitrary graphs in $O(D\Delta(D+\Delta))$ time  and in arbitrary trees in $O(D^2)$ time in the rooted initial configurations. The general initial configurations introduced $\max(D,k)$ factor instead of $O(D)$ factor in the time bounds for both arbitrary graphs and trees.

\vspace{1mm}
\noindent{\bf Related Work.}
\label{section:related}
There are three previous studies focusing on {\dis} in the local communication model.
Augustine and Moses Jr.~\cite{Augustine:2018} studied {\dis} 
assuming $k=n$. 
They proved a memory lower bound of $\Omega(\log n)$ bits at each robot and a time lower bound of $\Omega(D)$ ($\Omega(n)$ in arbitrary graphs) for any deterministic algorithm in any graph. 
They then 
provided deterministic algorithms using $O(\log n)$ bits at each robot to solve {\dis} on lines, rings, and trees in $O(n)$ time. For arbitrary graphs, they provided two algorithms, one using $O(\log n)$ bits at each robot with $O(mn)$ time and another using $O(n\log n)$ bits at each robot with $O(m)$ time. 

Kshemkalyani and Ali \cite{Kshemkalyani} provided an $\Omega(k)$ time lower bound for arbitrary graphs for $k\leq n$. 
They then provided three deterministic algorithms for {\dis} in arbitrary graphs: (i) The first algorithm using $O(k\log \Delta)$ bits at each robot with $O(m)$ time, (ii) The second algorithm using $O(D\log \Delta)$ bits at each robot with $O(\Delta^D)$ time, and (iii) The third algorithm using $O(\log(\max(k,\Delta)))$ bits at each robot with $O(mk)$ time. 
Recently, Kshemkalyani {\it et al.} \cite{Kshemkalyani2019} provided two algorithms: (i) the algorithm for arbitrary graph runs in $O(\min(m,k\Delta)\cdot \log k)$ time using $O(\log(\max(k,\Delta)))$ bits memory at each robot and (ii) the algorithm for grid graph runs in $O(\min(k,\sqrt{n}))$ time using $O(\log k)$ bits memory at each robot. 
Randomized algorithms are presented in  \cite{tamc19} to solve {\dis} where the random bits are mainly used to reduce the memory requirement at each robot. 
In this paper,
we
present results in the global communication model.
The previous results on arbitrary graphs and trees are summarized in Table \ref{table:comparision}. 

One problem that is closely related to {\dis} is the graph exploration by mobile robots. The exploration problem has been quite heavily studied in the literature for specific as well as arbitrary graphs, 
e.g., \cite{Bampas:2009,Cohen:2008,Dereniowski:2015,Fraigniaud:2005,MencPU17}. It was shown that a robot can explore an anonymous graph using $\Theta(D\log \Delta)$-bits memory; the runtime of the algorithm is $O(\Delta^{D+1})$ \cite{Fraigniaud:2005}. In the model where graph nodes also have memory, 
Cohen {\it et al.} \cite{Cohen:2008} gave two algorithms: The first algorithm uses $O(1)$-bits at the robot and 2 bits at each node, and the second algorithm uses $O(\log \Delta)$ bits at the robot and 1 bit at each node. The runtime of both algorithms is $O(m)$ with preprocessing time of $O(mD)$. The trade-off between exploration time and number of robots is studied in \cite{MencPU17}. 
The collective exploration by a team of robots is studied in \cite{FraigniaudGKP06} for trees. 
Another problem related to {\dis} is the scattering of $k$ robots in an $n$-node graph. This problem has been studied for rings \cite{ElorB11,Shibata:2016} and grids \cite{Barriere2009}. 
Recently, Poudel and Sharma \cite{Poudel18} provided a $\Theta(\sqrt{n})$-time algorithm for uniform scattering in a grid \cite{Das16}. 
Furthermore, {\dis} is  related to the load balancing problem, where a given
load at the nodes  has to be (re-)distributed among several processors (nodes). This problem has been studied quite heavily in graphs 
\cite{Cybenko:1989,Subramanian:1994}.
We refer readers to 
\cite{Flocchini2012,flocchini2019} for other recent developments in these topics.

\vspace{1mm}
\noindent{\bf Paper Organization.}  We discuss 
details of the model and some lower bounds in Section \ref{section:model}. 
We discuss the DFS traversal of a graph in Section \ref{section:basic}. 
We present a DFS-based algorithm for arbitrary graphs in Section \ref{section:algorithm0}, proving Theorem \ref{theorem:0}.  We then present a BFS-based algorithm for arbitrary graphs in Section \ref{section:algorithm1}, proving Theorem \ref{theorem:1}.
We then present a BFS-based algorithm for arbitrary trees in Section \ref{section:algorithm2}, proving Theorem \ref{theorem:2}.
Finally, we conclude in Section \ref{section:conclusion} with a short discussion. 
\onlyShort{Due to space constraints, one algorithm and many proofs are deferred to the full version attached in Appendix. }

\section{Model Details and Preliminaries}
\label{section:model}
\noindent{\bf Graph.} We consider the same graph model as in \cite{Augustine:2018,Kshemkalyani}. Let $G=(V,E)$ be an $n$-node graph with $m$ edges, i.e., $|V|=n$ and $|E|=m$. $G$ is assumed to be connected, unweighted, and undirected. 
$G$ is {\em anonymous}, i.e., nodes do not have identifiers but, at any node, its incident edges are uniquely identified by a {\em label} (aka port number) in the range $[1,\delta]$, where $\delta$ is the {\em degree} of that node. 
The {\em maximum degree} of $G$ is $\Delta$, which is the maximum among the degree $\delta$ of the nodes in $G$.
We assume that there is no correlation between two port numbers of an edge. 
Any number of robots are allowed to move along an edge at any time. 
The graph nodes do not have memory, i.e., they are not able to store any information. 


\vspace{1mm}
\noindent{\bf Robots.} We also consider the same robot model as in \cite{Augustine:2018,Kshemkalyani,Kshemkalyani2019}. Let $\cR=\{r_{1}, r_{2},\ldots,r_{k}\}$ be a set of $k\leq n$ robots residing on the nodes of $G$. For simplicity, we sometime use $i$ to denote robot $r_i$. No robot can reside on the edges of $G$, but one or more robots can occupy the same node of $G$.  
%
%
Each robot has a unique $\lceil \log k\rceil$-bit ID taken from $[1,k]$. 
When a robot moves from node $u$ to node $v$ in $G$, it is aware of the port of $u$ it used to leave $u$ and the port of $v$ it used to enter $v$. 
%
%
Furthermore, it is assumed that each robot is equipped with memory to store  information, which may also be read and modified by other robots present on the same node.

\vspace{1mm}
\noindent{\bf Communication Model.}
We assume that robots follow the global communication model, i.e., a robot is capable to communicate with any other robot in the system, irrespective of their positions in the graph nodes. However, they will not have the position information as graph nodes are anonymous. This is in contrast to the local communication model where  a robot can only communicate with other robots present on the same node.

\vspace{1mm}
\noindent{\bf Time Cycle.}
At any time a robot $r_i\in \cR$ could be active or inactive. When a robot $r_i$ becomes active, it performs
the ``Communicate-Compute-Move'' (CCM) cycle as follows. 
\begin{itemize}
\item 
{\em Communicate:} For each robot $r_j\in \cR$ that is at node some node $v_j$, a robot $r_i$ at node $v_i$ can observe the memory of $r_j$.  
Robot~$r_i$ can also observe its own memory. 
\item {\em Compute:} $r_i$ may perform an arbitrary computation
using the information observed during the ``communicate'' portion of
that cycle. This includes determination of a (possibly) port to 
use to exit $v_i$ and the information to store in 
the robot $r_j$ that is at $v_i$.
\item
{\em Move:} At the end of the cycle, $r_i$ writes new information (if any) in the memory of a robot $r_k$ at $v_i$,  and exits $v_i$ using the computed port to reach to a neighbor of $v_i$. 
\end{itemize}

\vspace{1mm}
\noindent{\bf Time and Memory Complexity.}
We consider the synchronous setting where 
every robot is active in every CCM cycle and they perform the cycle in a synchrony.
Therefore, time is measured in {\em rounds} or {\em steps} (a cycle is a round or step).
%
Another important parameter is memory. Memory comes from a single source -- the number of bits stored at each robot. 
%
%

\vspace{1mm}
\noindent{\bf Mobile Robot Dispersion.} The {\dis} problem can be formally defined as follows.
\begin{definition} [{\dis}]
Given any $n$-node anonymous graph $G=(V,E)$ having $k\leq n$ mobile robots positioned initially arbitrarily on the nodes of $G$, the robots reposition autonomously to reach a configuration where each robot is on a distinct node of $G$.
\end{definition}

The goal is to solve {\dis} optimizing two performance metrics: 
(i) {\bf Time} -- the number of rounds (steps), and
(ii) {\bf Memory} -- the number of bits stored at each robot.

\subsection{Some Lower Bounds}
We discuss here some time and memory lower bounds in the global communication model, which show the difficulty in obtaining fast runtime and lower memory algorithms. 
Consider the case of any rooted initial configuration of $k=n$ robots on a single node $v_{root}$ of an arbitrary graph $G$ with diameter $D$.
A time lower bound of $\Omega(D)$ is immediate since a robot initially at $v_{root}$ needs to traverse $\Omega(D)$ edges (one edge per time step) to reach a node that is $\Omega(D)$ away from $v_{root}$. 
For $k\leq n$, we present the following lower bound. 

\begin{theorem}
Any deterministic algorithm for {\dis} on graphs requires
$\Omega(k)$ steps in the global communication model.
\end{theorem}
\onlyLong{
\begin{proof}
Consider a line graph $G$ and a rooted initial configuration of $k\leq n$ robots on a single node $v_{root}$ of $G$.
In order for the robots to solve {\dis}, they need to dock at $k$ distinct nodes of $G$, exactly one on each node. To reach a node to dock, some robot must travel $k - 1$
edges of $G$, taking $k-1$ time steps.
\end{proof}
}

For $k=n$, we present the following time lower bound for trees. 

\begin{theorem}
For $k=n$, there exists a tree $T$ with $n$ nodes and diameter (height) $D$ such that any deterministic algorithm for {\dis} requires $\Omega(D^2)$ steps  in the global communication model. 
\end{theorem}
\onlyLong{
\begin{proof}
We use the lower bound proof for exploration due to Disser {\it et al.}~\cite{DisserMNSS16} to prove this lower bound. It has been argued in \cite{Augustine:2018} that a lower bound for exploration applies to {\dis}. 
We argue here that the lower bound of \cite{DisserMNSS16} applies for  {\dis} in the global communication model. 
Disser {\it et al.} \cite{DisserMNSS16} proved a lower bound for exploration assuming rooted initial configuration of $k=n$ robots are on a single node $v_{root}$ of tree $T$. Moreover, they assumed that the nodes of tree $T$ have unique identifiers and the robots have global communication. Specifically, they showed that: Using $k = n$
robots, there exists a tree $T$ on $n$ vertices and with diameter (height) $D = \omega(1)$ such that any deterministic exploration strategy 
requires at least $D^2/3=\Omega(D^2)$ steps to explore $T$. As our model is weaker because the nodes are indistinguishable, the $\Omega(D^2)$ steps lower bound applies to {\dis} in trees in the global communication model.   
\end{proof}
}

We finally prove a lower bound of $\Omega(\log (\max(k,\Delta)))$ bits at each robot  for any deterministic algorithm for {\dis} on graphs.

\begin{theorem}
\label{theorem:lower}
Any deterministic algorithm for {\dis} on $n$-node anonymous graphs requires  $\Omega(\log (\max(k,\Delta)))$ bits at each robot in the global communication model, where $k\leq n$ is the number of robots and $\Delta$ is  the maximum degree.
\end{theorem}
\onlyLong{
\begin{proof}
The memory lower bound of $\Omega(\log k)$ bits at each robot is immediate. Consider a rooted initial configuration of all $k$ robots on a single node $v_{root}$ of $G$. Since all robots run the same
deterministic algorithm, all the robots perform the  same moves. Therefore, arguing similarly as in Augustine and Moses Jr.~\cite{Augustine:2018}, we have $\Omega(\log k)$ bits memory lower bound at each robot. The memory lower bound of $\Omega(\log \Delta)$ bits at each robot
comes into play when $\Delta>k$. In this situation, to correctly recognize $\Delta$ different ports of a maximum degree node $v$, a robot needs $\Omega(\log \Delta)$ bits. Otherwise, the robot may not be able to move to all the neighbors of $v$, and hence {\dis} may not be achieved.
The global communication model does not help since 
the robots cannot differentiate which nodes next to its ports are already visited and which are not. 
\end{proof}
}

\section{DFS traversal of a Graph (Algorithm $DFS(k)$)}
\label{section:basic}
Consider an $n$-node arbitrary graph $G$ as defined in Section \ref{section:model}. Let $C_{init}$ be the initial configuration of $k\leq n$ robots positioned on a single node, say $v$, of $G$. 
Let the robots on $v$ be represented as $N(v)=\{r_1,\ldots, r_k\}$, where $r_i$ is the robot with ID $i$.  
We describe here a DFS traversal algorithm, $DFS(k)$, that disperses all the robots on the set $N(v)$ to the $k$ nodes of $G$ guaranteeing exactly one robot on each node. 
$DFS(k)$ 
will be heavily used in Section \ref{section:algorithm0} as a basic building block.

Each robot $r_i$ stores in its memory four variables $r_i.parent$ (initially assigned $null$), $r_i.child$ (initially assigned $null$), $r_i.treelabel$ (initally assigned $\top$), and $r_i.settled$ (initially assigned 0).  
%
%
%
 %
$DFS(k)$ executes in two phases, $forward$ and $backtrack$ \cite{Cormen:2009}. 
Variable $r_i.treelabel$ stores the ID of the smallest ID robot. Variable $r_i.parent$ stores the port from which $r_i$ entered the node where it is currently positioned in the forward phase.  
Variable $r_i.child$ stores the smallest port of the node  it is currently positioned at that has not been taken yet (while entering/exiting the node). 

We are now ready to describe $DFS(k)$. 
In round 1,
the maximum ID robot $r_k$ writes  $r_k.treelabel \leftarrow 1$ (the ID of the smallest robot in $N(v)$, which is 1),  $r_k.child\leftarrow 1$ (the smallest port at $v$ among $P(v)$), 
and $r_k.settled\leftarrow 1$. 
The robots $N(v)\backslash\{r_k\}$ exit $v$ following port $r_k.child$; $r_k$ stays (settles) at $v$. 
%
In the beginning of round 2, the robots $N(w)=N(v)\backslash\{r_k\}$ reach a neighbor node $w$ of $v$.
Suppose the robots entered $w$ using port $p_w\in P(w)$.
As $w$ is free, robot $r_{k-1}\in N(w)$ writes $r_{k-1}.parent\leftarrow p_w$, $r_{k-1}.treelabel \leftarrow 1$ (the ID of the smallest robot in $N(w)$), and $r_{k-1}.settled\leftarrow 1$.
If $r_{k-1}.child\leq\delta_w$, $r_{k-1}$ writes
$r_{k-1}.child\leftarrow r_{k-1}.child+1$ if port $r_{k-1}.child+1\neq p_w$ and $r_{k-1}.child+1\leq \delta_w$, otherwise $r_{k-1}.child\leftarrow r_{k-1}.child+2$. 
The robots $N(w)\backslash\{r_{k-1}\}$ decide to continue DFS in forward or backtrack phase as described below.

\begin{itemize}
\item ({\bf forward phase}) if ($p_w=r_{k-1}.parent$ or $p_w=$ old value of $r_{k-1}.child$) and (there is (at least) a port at $w$ that has not been taken yet).
The robots $N(w)\backslash\{r_{k-1}\}$ exit $w$ through port $r_{k-1}.child$.  

\item ({\bf backtrack phase}) if ($p_w=r_{k-1}.parent$ or $p_w=$ old value of $r_{k-1}.child$) and (all the ports of $w$ have been taken already).
The robots $N(w)\backslash\{r_{k-1}\}$ exit $w$ through port $r_{k-1}.parent$.
\end{itemize}

Assume that in round 2, the robots decide to proceed in forward phase. 
In the beginning of round 3, 
$N(u)=N(w)\backslash\{r_{k-1}\}$ robots reach some other node $u$ (neighbor of $w$) of $G$. The robot $r_{k-2}$ stays at $u$ writing necessary information in its variables. In the forward phase in round 3, the robots $N(u)\backslash\{r_{k-2}\}$ exit $u$ through port $r_{k-2}.child$.
However, in the backtrack phase in round 3, $r_{k-2}$ stays at $u$ and robots $N(u)\backslash\{r_{k-2}\}$ exit $u$ through port $r_{k-2}.parent$. This takes robots $N(u)\backslash\{r_{k-2}\}$ back to node $w$ along $r_{k-1}.child$. Since $r_{k-1}$ is already at $w$, $r_{k-1}$ updates $r_{k-1}.child$ with the next port to take.
Depending on whether $r_i.child\leq \delta_w$ or not,
the robots $\{r_1,\ldots,r_{k-3}\}$ exit $w$ using either $r_{k-1}.child$ (forward phase) or $r_{k-1}.parent$ (backtrack phase). 

There is another condition, denoting the onset of a cycle, under which choosing backtrack phase is in order. When the robots enter $x$ through $p_x$ and robot $r$ is settled at $x$, 
\begin{itemize}
\item ({\bf backtrack phase}) if ($p_x\neq r.parent$ and $p_x\neq$ old value of $r.child$). The robots exit $x$ through port $p_x$ and no variables of $r$ are altered.
\end{itemize}
This process then continues for $DFS(k)$ until at some node $y\in G$, $N(y)=\{r_1\}$. The robot $r_1$ then stays at $y$ and $DFS(k)$ finishes. 

\begin{lemma}
Algorithm $DFS(k)$ correctly solves {\dis} for  $k\leq n$ robots initially positioned on a single node of a $n$-node arbitrary graph $G$  in $ \min(4m-2n+2,  k\Delta)$ rounds using $O(\log(\max(k,\Delta)))$ bits at each robot.
\label{dfscorrect}
\end{lemma}
\onlyLong{
\begin{proof}
We first show that {\dis} is achieved by $DFS(k)$. Because every robot starts at the same node and follows the same path as other not-yet-settled robots until it is assigned to a node, $DFS(k)$ resembles the DFS traversal of an anonymous port-numbered graph \cite{Augustine:2018} with all robots starting from the same node. 
Therefore, $DFS(k)$ visits $k$ different nodes where each robot is settled. 


We now prove time and memory bounds. In $k \Delta$ rounds, $DFS(k)$ visits at least $k$ different nodes of $G$. If $4m-2n+2<  k \Delta$, $DFS(k)$ visits all $n$ nodes of $G$. Therefore, it is clear that the runtime of $DFS(k)$ is $ \min(4m-2n+2, k\Delta)$ rounds. Regarding memory, variable $treelabel$  takes $O(\log k)$ bits, $settled$ takes $O(1)$ bits, and $parent$ and $child$ take $O(\log \Delta)$ bits.
The $k$ robots can be distinguished through $O(\log k)$ bits since their IDs are in the range $[1,k]$.
Thus, each robot requires $O(\log(\max(k,\Delta)))$ bits.  
\end{proof}
}

\section{Algorithm for Arbitrary Graphs (Theorem \ref{theorem:0})}
\label{section:algorithm0}
We present and analyze {\em Graph\_Disperse\_DFS}, a DFS-based algorithm that solves {\dis} of $k\leq n$ robots on an arbitrary $n$-node graph in $O(\min(m,k\Delta ))$ time with $O(\log (\max(k,\Delta)))$ bits of memory at each robot in the global communication model. 
This algorithm improves the $O(\min(m,k\Delta)\cdot \log k)$ time of the best previously known algorithm \cite{Kshemkalyani} for arbitrary graphs (Table \ref{table:comparision}) in the local communication model.

\begin{algorithm}[!h]
{\small
Initialize: $TID\_Grow$, $TID\_Collect$, $CID$, $CID\_old$, $winner$, $leader$, $home$, $state$\\
\If{$state=grow$}{ 
 \If{node is free}{
  highest ID robot $x$ from highest $CID$ group having $state=grow$ settles; $x.state\gets settled$
 }
 \If{$CID\neq r.CID$}{
  one such robot from each $CID$ group broadcasts $Subsume(CID,r.CID)$
 }
 Subsume\_Graph\_Processing\\
 \If{$CID$ is node in Subsume graph}{
  $CID\_old\gets CID$; $CID\gets$ $winner$ in my component of Subsume graph\\
  Let $x\gets \min_{j}(j.TID\_grow=winner \wedge j.state=grow \wedge j \mbox{ is at same node as } i)$\\
  $x.home\gets r.ID$; $x.leader\gets 1$; $x.state\gets collect$; $x.TID\_Collect\gets x.TID\_Grow$\\
  $x$ begins DFS $Collect(TID\_Collect)$\\
  \If{$i\neq x \wedge state=grow$}{
   $state\gets subsumed$; STOP
  }
 }
 \Else{
  continue DFS $Grow(TID\_Grow)$
 }
}
\ElseIf{$state=collect$}{
 Subsume\_Graph\_Processing\\
 \If{$CID$ is node in Subsume graph}{
  $CID\_old\gets CID$; $CID\gets$ $winner$ in my component of Subsume graph\\
  $state\gets subsumed$; STOP\\
  \If{$leader=1$}{
   $leader\gets 0; home\gets\perp$
  }
 }
 \Else{
  \If{node is free $\vee CID=r.CID=r.CID\_old \vee CID\neq r.CID$}{
   backtrack, as part of DFS $Collect(TID\_Collect)$
  }
  \ElseIf{$CID=r.CID\neq r.CID\_old$}{
   \If{$\exists x\,|\, x.leader=1 \wedge x.home=r.ID \wedge$ all ports at $r$ have been explored}{
    \If{$x=i$}{
     $x.leader\gets 0; x.home\gets\perp$
    }
    $state\gets grow$; $TID\_Grow\gets x.TID\_Grow$\\
    $i$ continues DFS $Grow(TID\_Grow)$
   }
   \Else{
    $i$ continues DFS $Collect(TID\_Collect)$ of $x \,|\,x.leader=1$, along with $x$ if $x$ is backtracking to its parent in DFS $Collect(x.TID\_Collect)$
   }
  }
 }
}
} 
\caption{Algorithm {\em Graph\_Disperse\_DFS} to solve \dis\ in global model. Code for robot $i$ in a round at any node. $r$ denotes a settled robot, if any, at that node.}
\label{algo1g}
\end{algorithm}
\setlength{\textfloatsep}{0pt}


\begin{algorithm}[t]
{\small
\setcounter{AlgoLine}{34}
\ElseIf{$state=subsumed$}{
 Subsume\_Graph\_Processing\\
 \If{$CID$ is node in Subsume graph}{
  $CID\_old\gets CID$; $CID\gets$ $winner$ in my component of Subsume graph
 }
 \Else{
  \If{$\exists$ arrived robot $x \,|\,x.state=collect \wedge x.leader=1$ $\wedge$ $x$ is backtracking to its parent in DFS $Collect(TID\_Collect)$}{
   $state\gets collect; TID\_Collect\gets x.TID\_Collect$\\
   \If{$x.home=r.ID \wedge $ all ports at $r$ have been explored}{
    $state\gets grow; TID\_Grow\gets x.TID\_Grow$\\
    $i$ continues DFS $Grow(TID\_Grow)$ along with $x$
   }
   \Else{ 
    $i$ continues DFS $Collect(TID\_Collect)$ along with $x$
   }
  }
 }
}
\ElseIf{$state=settled$}{
 Subsume\_Graph\_Processing\\
 \If{$CID$ is node in Subsume graph}{
  $CID\_old\gets CID$; $CID\gets$ $winner$ in my component of Subsume graph
 }
}
\underline{Subsume\_Graph\_Processing}\\
receive Subsume messages; build Subsume graph $S$\\
\If{node with no incoming edge in my component of $S$ exists}{
 $winner\gets \min_{CID}$($CID$ of nodes with no incoming edge in my component of $S$)
}
\Else{
 $winner\gets$ $\min_{CID}$(cycle of CIDs in my component)
}
}

\label{algo1g-p2}
\end{algorithm}

\subsection{The Algorithm}
The algorithm is based on DFS traversal. In general, a robot may operate in one of two interchangeable phases: GROW and COLLECT. As these are independent, a separate set of DFS variables: $parent$, $child$, is used for operating in the two phases. The following additional variables are used.
(i) $TID\_Grow$: Tree ID, of type robot identifier, is the ID of the DFS tree in the GROW phase with which the robot is associated. Initially, $TID\_Grow\gets$ minimum ID among the colocated robots.
(ii) $TID\_Collect$: Tree ID, of type robot identifier, is the ID of the DFS tree in the COLLECT phase with which the robot is associated. Initially, $TID\_Collect\gets\perp$.
(iii) $CID$: for component ID, of type robot identifier, is used to denote the component associated with the GROW phase of the DFS. Initially, $CID\gets TID\_Grow$.
(iv) $CID\_old$: for earlier component ID, of type robot identifier, is used to denote the earlier value of component ID just before the most recent component ID ($CID$) update, associated with the GROW phase of the DFS. Initially, $CID\_old\gets TID\_Grow$.
(v) $winner$: of type robot identifier. When multiple components collide/merge, this is used to indicate the winning component ID that will subsume the other components. Initially, $winner\gets\perp$. 
(vi) $leader$: of type boolean. This is set to 1 if the robot is responsible for collecting the various robots distributed in the component. Initially, $leader\gets 0$.
(vii) $home$:  of type robot identifier. The robot identifier of a settled robot is used to identify the origin node of the leader robot that is responsible for collecting the scattered robots in the component back to this origin node. Initially, $home\gets\perp$.
(viii) $state$: denotes the state of the robot and can be one of $\{grow, collect, subsumed, settled\}$. Initially, $state\gets grow$.

In the initial configuration, there are groups of robots at different nodes. Each robot has its $TID\_Grow$ set to the minimum ID among the colocated robots, and its $state=grow$. The robots from a node move together in a DFS traversal, to locate free nodes and settle one by one. As they do the DFS traversal $Grow(TID\_Grow)$, they extend the DFS tree that is associated with the $TID\_Grow$. Each growing DFS tree is also associated with a component ID, $CID$, that is initialized to the $TID\_Grow$. Multiple DFS trees associated with different $CID$s may meet at a node in any round; specifically, a DFS tree for component $CID$ may meet another component $r.CID$ for some other DFS tree, where $r$ is the robot that is settled at that node. In this case, one robot from the newly arrived robots of the DFS tree component $CID$ broadcasts a $Subsume(CID,r.CID)$ message. This is to indicate that the component $CID$ is subsuming the component $r.CID$. Multiple such $Subsume$ messages may get broadcast from different robots in the graph in any particular round. 

All the robots listen to all such broadcasts in each round, and build a directed graph, $Subsume$, $S=(C,E)$, where $C$ is the set of component IDs, and edge $(CID_j,CID_k)$ indicates that $Subsume(CID_j,CID_k)$ message has been received. In this graph, each node may have at most one outgoing edge but may have multiple incoming edges. The $winner$ component ID corresponds to that node (in my connected component of $S$) that has the minimum $CID$ among the nodes with no incoming edges (if such a node exists). Otherwise, there must exist a cycle with no incoming edges in the $Subsume$ graph, and the lowest valued $CID$ node in the cycle is chosen as $winner$. The significance of the $winner$ is that its $CID$ subsumes all other $CID$s in its connected component of $S$; that is, all robots that are in the same connected component of $S$ overwrite their current $CID$ by $winner$ in their connected component of $S$. 

The robot with the minimum ID among those with $TID\_Grow = winner$ and $state=grow$ changes its $state$ to $collect$, $leader$ to 1, $TID\_Collect$ to $TID\_Grow$, and embarks on the $Collect$ phase. In the $Collect$ phase, the leader does an independent DFS traversal $Collect(TID\_Collect)$ of the connected component of settled nodes of $G$ which have settled robots which have newly changed their component ID $CID$ to be the same as its own. And all (unsettled) robots which have newly changed their $CID$ to that of the $winner$ leader, or have $CID=winner$ but are not the leader, also change their $state$, whether $grow$ or $collect$, to $subsumed$ and stop movement until they are collected. In this DFS traversal $Collect$, the leader node collects all unsettled robots with $state=subsumed$ and brings them back to its home node from where it began the $Collect$ DFS traversal, while the thus collected robots change their $state$ to $collect$ once they join the collection traversal. During the $Collect$ traversal, if in some step the component gets subsumed by some other component, the unsettled robots reset their $state$ to $subsumed$. If the $Collect$ DFS traversal completes successfully, the collected robots and the leader change $state$ to $grow$, set their $TID\_Grow$ to that of the leader, and resume DFS $Grow(TID\_Grow)$ after the leader resets its leader status.

Note that the DFS $Collect(TID\_Collect)$ is independent of the DFS $Grow(TID\_Grow)$, and thus an independent set of variables $parent$, $child$, need to be used for the two different types of DFSs. Further, when a new instance of a DFS $Grow$/DFS $Collect$, as identified by a new value of $TID\_Grow$/ $TID\_Collect$, is detected at a settled robot (node), the robot switches to the new instance and resets the old values of $parent$ and $child$ for that DFS search.

In the DFS $Collect$ phase, the leader visits all nodes in its connected component of settled nodes having a settled robot that changed its component ID $r.CID\gets CID$. (These are the settled robots where $CID=r.CID\neq r.CID\_old$, where $CID\_old$ is the value of $CID$ before the latest overwrite by $winner$.)
This excludes the nodes already visited in the DFS $Grow$ phase having settled robots with the same $CID$ as that of the leader before it become the leader.
To confine the DFS $Collect$ to such nodes, note that the leader may have to backtrack from a node $v$ if the node (i) is free or (ii) has $CID=r.CID=r.CID\_old$ or (iii) has $CID\neq r.CID$. If the $CID$ of the leader changes at the beginning of this round (because it gets subsumed), before 
it can backtrack, the leader (and any accompanying robots having $state=collect$) simply changes $state$ to $subsume$ and stops. In cases (i) and (iii), there may thus be stopped robots at a free node, or at a node that belongs to an adjoining, independent component. Such robots may be later collected by (a) a leader from its old component, or (perhaps earlier than that) (b) by a leader from the component where they stop. In the former case (a), it is execution as usual. In the latter case (b), there is no issue of violating correctness even though the robots jump from one connected component sharing a common $CID$ to an adjacent one with a different $CID$. 

\subsection{Correctness and Complexity}
A robot may be in one of four states: $grow$, $collect$, $subsumed$, and $settled$. The state transition diagram for a robot is shown in Figure~\ref{fig:stdcid}(a).

\begin{figure*}[!t] 
\centering
\includegraphics[height=1.2in]{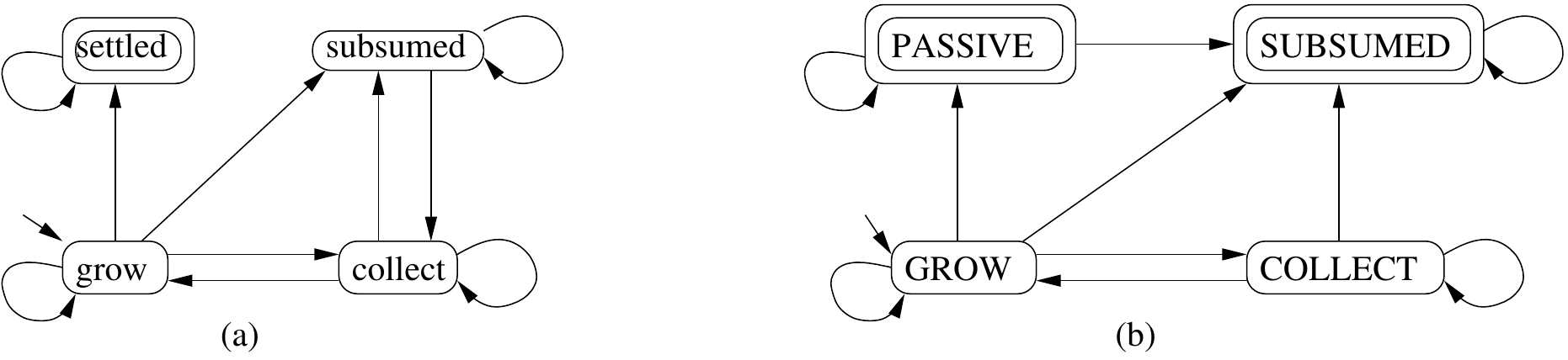}
\caption{State transition diagrams. (a) Diagram for a robot's state, $state$. (b) Diagram for any value of $CID$.}
\label{fig:stdcid}
\end{figure*}

\begin{lemma}\label{growsteps}
Once a robot enters $state=grow$ for some value of $TID\_Grow$, the DFS $Grow(TID\_Grow)$ completes within $\min(4m-2n+2,4k\Delta)$ rounds, or the robot moves out of that state within $\min(4m-2n+2,4k\Delta)$ rounds.
\end{lemma}
\begin{proof}
A DFS (using $\log \Delta$ bits at a robot) completes within $4m-2n+2$ rounds. It also completes within $4k\Delta$ rounds, as in the DFS, each edge gets traversed a maximum of 4 times, and as at most $k$ nodes need to be visited before all $k$ robots get settled. Before this completion of the DFS, if the current component gets subsumed or subsumes another component, the robot moves to either $subsumed$ or $collect$ state.
\end{proof}

\begin{lemma}\label{collectsteps}
Once a robot enters $state=collect$ for some value of $TID\_Collect$, the DFS $Collect(TID\_Collect)$ completes in $\min(4m-2n+2,4k\Delta)$ rounds or the robot moves out of that state within $\min(4m-2n+2,4k\Delta)$ rounds.
\end{lemma}

\begin{proof}
A DFS traversal of a component (using $\log \Delta$ bits at a robot) completes within $4m-2n+2$ rounds. It also completes within $4k\Delta$ rounds, as the collecting robot in the DFS traverses an edge at most 4 times, needs to visit each of the at most $\Delta$ neighbors of the at most $k$ settled nodes in the component, until collection completes and the leader is back at the home node. (At the completion of the DFS, the robot moves to $grow$ state; before this completion of the DFS, if the current component gets subsumed by another component, the robot moves to $subsumed$ state.)
\end{proof}

\begin{theorem}\label{th:disp}
The algorithm $Graph\_Disperse\_DFS$ solves \dis.
\end{theorem}
\begin{proof}
Each robot begins in $state=grow$. We make the following observations about the state transition diagram of a robot.
\begin{enumerate}
    \item A robot can enter $subsumed$ state at most $k-1$ times. In $subsumed$ state, a robot can stay at most $\min(4m-2n+2, 4k\Delta)\cdot k$ rounds before it changes $state$ to $collect$.
    \item From $collect$ state, within $\min(4m-2n+2,4k\Delta)$ rounds, a robot can go to $subsumed$ state (which can happen at most $k-1$ times), or go to $grow$ state (Lemma~\ref{collectsteps}). 
    \item In $grow$ state, a robot can remain for at most $\min(4m-2n+2,4k\Delta)$ rounds, by when it may go to either $subsumed$ or $collect$ state for at most $k-1$ times, or go to $settled$ state 
    (Lemma~\ref{growsteps}). 
\end{enumerate}
It then follows that within a finite, bounded number of rounds, a robot will be in $grow$ state for the last time and within $\min(4m-2n+2,4k\Delta)$ further rounds, settle and go to $settled$ state. This is and will be the only robot in $settled$ state at the node. Thus, \dis\ is achieved within a finite, bounded number of rounds. 
\end{proof}

We model the state of a particular value of $CID$. It can be either GROW, COLLECT, PASSIVE, or SUBSUMED. The state transition diagram for the state of a $CID$ value is shown in Figure~\ref{fig:stdcid}(b).

\begin{theorem}\label{th:term}
The algorithm $Graph\_Disperse\_DFS$ terminates in $\min(2\cdot 4m,2\cdot 4k\Delta)$ rounds.
\end{theorem}

\onlyLong{
\begin{proof}
By Theorem~\ref{th:disp}, each robot settles within a finite, bounded number of rounds. We now determine the number of rounds more precisely.

Let $CID_x$ denote the $CID$ of any robot (one of possibly several) that settles in the last round of Algorithm $Graph\_Disperse\_DFS$. This $CID_x$ has never been subsumed and therefore its state has shuttled between GROW and COLLECT before reaching and ending in PASSIVE. We separately bound the number of rounds spent by $CID_x$ in GROW state and in COLLECT state.

Let the DFS tree in GROW state be associated with $TID\_Grow_x$. Observe that multiple sojourns of $CID_x$ in GROW state are associated with the same $TID\_Grow_x$. The DFS data structures associated with $TID\_Grow_x$ are never overwritten by another DFS in GROW state as the component $CID_x$ is never subsumed (and independent DFS traversal data structures are maintained for the GROW and COLLECT phases). 
Within $4m-2n+2$ rounds, possibly spread across multiple sojourns in GROW state, the DFS associated with $TID\_Grow_x$ completes and every robot associated with it gets settled. Every robot associated with $TID\_Grow_x$ also gets settled within $4k\Delta$ rounds, as the DFS visits each edge at most 4 times, and hence within $4k\Delta$ rounds, at least $k$ nodes get visited.

$CID_x$ can transit from GROW to COLLECT and back at most $k-1$ times because that is the maximum number of times $CID_x$ can subsume another $CID$. Let the transition to COLLECT state occur $l, 0\leq l\leq k-1$ times, let the number of rounds spent in COLLECT state on the $j$th transition to it, $1\leq j \leq l$, be $r_j$. Each transition to COLLECT is followed by a successful DFS traversal of the component $C_j$ of nodes having $CID=CID_x\neq CID\_old$.
There are four types of edges traversed in the DFS traversal of $C_j$.
\begin{enumerate}
    \item $e^{int}_j$: Edge between two nodes in $C_j$. 
    \item $e^{back}_j$: Edge from a node in $C_j$ to a node having $CID=CID_x=CID\_old$.
    \item $e^{free}_j$: Edge from a node in $C_j$ to a free node.
    \item $e^{adj}_j$: Edge from a node in $C_j$ to a node in an adjacent component (having $CID_y\neq CID_x$).
\end{enumerate}
Let the sets of these four types of edges be denoted $E^{int}_j$, $E^{back}_j$, $E^{free}_j$, $E^{adj}_j$.
We observe the following constraints from the algorithm.
\begin{enumerate}
\item Edge $e^{int}_j$ will never be incident on any node in any component $C_{j'}$, for $j'>j$. 
\item Similarly, edge $e^{back}_j$ will never be incident on any node in any component $C_{j'}$, for $j'>j$. 
\item Edge $e^{free}_j = (u,v)$ may at most once become an edge $e^{back}_{j'} = (v,u)$, for $j'>j$, if the free node $v$ gets settled and the component it is in gets subsumed by $CID_x$. 
\item Similarly, edge $e^{adj}_j = (u,v)$ may at most once become an edge $e^{back}_{j'} = (v,u)$, for $j'>j$, if the component that node $v$ is in gets subsumed by $CID_x$.
\end{enumerate}
In the DFS of $C_j$, each $e^{int}_j$ is traversed at most 4 times, whereas each $e^{back}_j$, $e^{free}_j$, and $e^{adj}_j$ is traversed at most 2 times (once in the forward mode and once in the backtrack mode). This gives:
\[r_j \leq 4|E^{int}_j| + 2(|E^{back}_j| + |E^{free}_j| + |E^{adj}_j|)\]
From the above constraints, we have:
\[\sum_{j=1}^l r_j \leq 4m\]

The DFS $Collect(TID\_Collect=CID_x)$ on the $j$th transition to COLLECT state contains a DFS traversal of the component $C_j$ of settled nodes having $CID=CID_x\neq CID\_old$; denote by $N^{int}_j$, the set of such nodes. For each $j, 1\leq j \leq l$, the set of such settled nodes is disjoint. In the DFS traversal of $C_j$, (the leader from) each node in $N^{int}_j$ visits each adjacent edge at most 4 times. As the number of such settled nodes across all $j$ is at most $k$, it follows that at most $4k\Delta$ edges are visited in DFS $Collect(TID\_Collect=CID_x)$ across all transitions to COLLECT state. Hence, \[\sum_{j=1}^l r_j \leq 4k\Delta\]

The theorem follows by separately combining the number of rounds in GROW and COLLECT phases in terms of $m$, and separately combining the number of rounds in GROW and COLLECT phases in terms of $k\Delta$.
\end{proof}
}

\begin{theorem}
Algorithm~\ref{algo1g} ({\em Graph\_Disperse\_DFS}) requires $O(\log(\max(k,\Delta)))$ bits memory.
\label{th:dfsspace}
\end{theorem}
\begin{proof}
Each set of $parent$, $child$, $treelabel$, $settled$ for the GROW and COLLECT phases takes $O(\log(\max(k,\Delta)))$ bits (follows from Theorem~\ref{dfscorrect}).
$CID$, $CID\_old$, $winner$, $TID\_Grow$, $TID\_Collect$, and $home$ take $O(\log\,k)$ bits each. $leader$ and $state$ take $O(1)$ bits each.
Thus, the theorem follows.
\end{proof}

{\bf Proof of Theorem~\ref{theorem:0}:} Follows from 
Theorems~\ref{th:disp} 
-- \ref{th:dfsspace}.

\section{Algorithm for Arbitrary Graphs (Theorem \ref{theorem:1})}
\label{section:algorithm1}

We present and analyze  {\em Graph\_Disperse\_BFS}, a BFS-based algorithm that solves {\dis} of $k\leq n$ robots on an arbitrary $n$-node graph in $O(\max(D,k)\Delta(D+\Delta))$ time with $O(\max(D,\Delta\log k))$ bits of memory at each robot in the global communication model. 
This algorithm improves the $O(\Delta^D)$ time of the best previously known algorithm \cite{Kshemkalyani} for arbitrary graphs (Table \ref{table:comparision}) in the local communication model.

We first discuss the rooted graph case, wherein all robots are on a single node initially. Here, we show that dispersion takes  $O(D\Delta(D+\Delta))$ time with $O(\max(D,\Delta\log k))$ bits of memory at each robot in the local communication model. We then extend the result to the general case of robots on multiple nodes.

\subsection{Rooted Case}
In the initial configuration, all $k\leq n$ robots are at a single node $v_{root}$. 
The synchronous algorithm proceeds in rounds and is described in Algorithm~\ref{algo:algo2bfs}. 
The algorithm induces a breadth-first search (BFS) tree, level by level, in the graph. 
There are two main steps in the algorithm when extending the tree from $i$ levels to $i+1$ levels: (i) the leaf nodes in the BFS tree at level $i$ determine the number of edges going to level $i+1$. This is done in procedure {\tt DetermineLeafDemand(i)} and can be achieved in $O(\Delta^2)$ rounds as a 2-neighborhood traversal is performed. The level $i$ robot sets its demand for robots equal to the number of edges (ports) going to level $i+1$. (ii) The leaf nodes at level $i$ then populate the level $i+1$ nodes in a coordinated manner, because there may be arbitrary number of edges and connectivity going from level $i$ to level $i+1$. This is done in procedure {\tt PopulateNextLevel(i)} iteratively by borrowing robots for exploration from $v_{root}$. The number of robots assigned by the root may be up to the demand of that node. This movement takes $O(D)$ time. As there may be edges from multiple nodes at level $i$ to a node at level $i+1$, only one robot can be earmarked to settle at that node; others retrace their steps back to the root. The robot earmarked to settle at the level $i+1$ node does a 1-neighborhood traversal and {\em invalidates} the ports of level $i$ nodes leading to that level $i+1$ node ($O(\Delta)$ time). The robot does not actually settle at the level $i+1$ node but participates in further computation. It then returns to the level $i$ node it arrived from and designates the port used to go to the level $i+1$ node as a {\em valid} port. The settled robots at level $i$ then re-evaluate the demand for robots, based on the number of {\em unfinalized} ports (i.e., not validated and not invalidated ports going to level $i+1$ nodes). All unsettled robots (including those that had been ``earmarked'' to settle at a level $i+1$ node) return to $v_{root}$ and they are reassigned for the next iteration based on the renewed (and decreased) valued of net demand for exploratory robots. This movement takes $O(D)$ time. The algorithm guarantees that $\Delta$ iterations suffice for determining the {\em valid/invalid} status of all ports of level $i$ nodes leading to level $i+1$ nodes, after which a final iteration reassigns the final demand based on the number of {\em valid} ports (each of which leads to a unique level $i+1$ node) and distributes up to those many robots among level $i+1$ nodes. The procedure {\tt PopulateNextLevel(i)} thus takes $O(\Delta(\Delta + D))$ time. 

Due to the BFS nature of the tree growth, $D$ iterations of the outer loop of {\em Graph\_Disp\-erse\_BFS} suffice. Hence, the running time is $O(D(\Delta^2+\Delta(\Delta+D))$. 

The following variables are used at each robot. (i) $nrobots$ for the total number of robots at the root, $v_{root}$. Initialize as defined. (ii) $level$ for the level of a robot/node in the BFS tree. Initialize to 0. (iii) $i$ for the current maximum level of settled robots. Initialize to 0. (iv) $demand[1\ldots\Delta]$, where $demand[j]$ for a non-leaf node is the demand for robots to populate level $i+1$ for sub-tree reachable via port $j$. Initialize to $\overline{0}$. $demand[j]$ for a leaf node at level $i$ has the following semantics. $demand[j] = 0/1/2/3$ means the node reachable via port $j$ does not go to level $i+1$ node/goes to a level $i+1$ node via an unfinalized edge/goes to a validated level $i+1$ node/goes to an invalidated level $i+1$ node. (v) $child\_id[1\ldots\Delta]$, where $child\_id[j]$ is the ID of the child node (if any,) reachable via port $j$. Initialize to $\overline{\perp}$. (vi) $parent\_id$ for the ID of the parent robot in the BFS tree. Initialize to $\perp$. (vii) $parent$ to identify the port through which the parent node in the BFS tree is reached. Initialize to $\perp$. (viii) $winner$ to uniquely select a robot among those that arrive at a level $i+1$ node. Initialize to 0.

The variables $child\_id[1\ldots\Delta]$ and $parent\_id$, and the unique robot identifiers assumption, are strictly not necessary for the single-rooted case. Without $child\_id[1\ldots\Delta]$ and $parent\_id$, the broadcast function can be simulated by each settled robot moving up to its parent and back, to communicate the demand of its subtree. The unique robot identifiers assumption and $winner$ help in determining which robot should settle at the root, for assigning robots as per the demands, and for selecting $winner$. Without these, a simple randomized scheme can be used for the above determinations.

\begin{lemma}
The {\bf while} loop of {\tt PopulateNextLevel(i)} (line (11) in Algorithm~\ref{algo:algo2bfs}) terminates within $\Delta$ iterations.
\label{lemma:bfsrootedbound}
\end{lemma}
\onlyLong{
\begin{proof}
Let $e_i$ denote $\sum_{leaf\,u}\sum_{j=1}^{\Delta} (1\,if\,demand_u[j]=1)$, the number of unfinalized edges (i.e., not validated and not invalidated edges) going from level $i$ to level $i+1$. 
\begin{enumerate}
\item If $nrobots\geq e_i$, then every edge in $e_i$ can be explored and all nodes in level $i+1$ accounted for by validating exactly one edge each among the ports at level $i$ nodes (and invalidating all other unfinalized edges at level $i$ nodes) -- thus, the {\bf while} loop can be exited after one iteration as $demand_u[j]\neq 1$ for any $u$ at level $i$ and for any $j$ and the first clause of the {\bf while} loop condition is falsified.
\item
If $nrobots<e_i$, then $nrobots$ robots will be pressed into service for exploration of level $i+1$, at least $\lceil nrobots/\Delta\rceil$ nodes at level $i+1$ will be visited uniquely in this iteration (i.e., not visited in earlier iterations) via unfinalized edges, and hence at least $\lceil nrobots/\Delta\rceil$ ports (edges) at level $i$, that are currently marked as $demand_u[j]=1$ will be validated with change $demand_u[j]=2$. For the next iteration of the {\bf while} loop, $e_i$ will be decreased by at least this amount. It follows that within $\Delta$ iterations, (i) at least $nrobots$ ports at level $i$ will be validated ($demand_u[j]=2$), or (ii) $nrobots\geq e_i$. In case (i), the second clause of the {\bf while} loop condition is falsified and the loop is exited. In case (ii), by the reasoning given above in part (1), in one additional iteration, the loop is exited. 
\end{enumerate}
\end{proof}
}

\begin{lemma}
A BFS tree in induced in the underlying graph by Algorithm {\em Graph\_Disperse\_BFS} given in Algorithm~\ref{algo:algo2bfs}.
\label{lemma:bfsrootedbfs}
\end{lemma}
\onlyLong{
\begin{proof}
We show by induction on the hypothesis that ``all nodes at distance $i$ (along shortest path) from the root have a settled robot that is assigned $level=i$, or there are no more robots to assign to some such nodes.''
The hypothesis is clearly true for $level=0$ and can be seen to be true for $level=1$ by following the execution of the algorithm.

We now assume the hypothesis for $level=x$ and prove it true for $level=x+1$. Procedure {\tt DetermineLeafDemand(x)} correctly identifies all nodes at level $x+1$ and the number of unfinalized edges going to such nodes from level $x$ nodes is set to $e_i$ = $\sum_{u\,at\,level\,x}\sum_{j=1}^{\Delta} (1\,if\,demand_u[j]=1)$.
\begin{enumerate}
\item
If $nrobots\geq e_i$, then in one iteration of the {\bf while} loop of {\tt PopulateNextLevel(x)}, all nodes of level $x+1$ are assigned robots with $level=x+1$ as $\sum_{u\,at\,j}\sum_{j=1}^{\Delta} (1\,if\,demand[j]=1) = 0$. Each node at level $x+1$ has its corresponding outgoing port $j$ from level $x$ parent $u$ set to $demand_u[j]=2$, with possibly more robots left at $v_{root}$ for populating higher levels.
\item
If $nrobots<e_i$, then (as argued in the proof of Lemma~\ref{lemma:bfsrootedbound}),
it follows that within $\Delta$ iterations of the {\bf while} loop of {\tt PopulateNextLevel(x)}, (i) at least $nrobots$ ports at level $i$ will be validated ($demand_u[j]=2$), or (ii) $nrobots\geq e_i$. 
In case (i), the second clause of the {\bf while} loop condition is falsified and the loop is exited. All the remaining robots ($nrobots$) can be accommodated at level $x+1$ nodes and possibly some nodes at level $x+1$ will not be assigned any robots because the algorithm has run out of robots. There will be no further levels in the BFS tree. 
In case (ii), by the reasoning given above in part (1), in one additional iteration, the loop is exited. All the nodes at level $x+1$ will be assigned robots with $level=x+1$, with possibly more robots left at $v_{root}$ for populating higher levels. Each node at level $x+1$ has its corresponding outgoing port $j$ from level $x$ parent $u$ set to $demand_u[j]=2$.
\end{enumerate}
The correctness of the induced BFS tree follows.
\end{proof}
}

\begin{theorem}
Algorithm~\ref{algo:algo2bfs} ({\em Graph\_Disperse\_BFS}) solves \dis\ on single-rooted graphs in 
$O(D\Delta(\Delta + D))$ rounds and requires $O(\max(D,\Delta\log\,k))$ memory in the local communication model.
\label{th:bfsrootedtime}
\end{theorem}
\begin{proof}
There is one robot settled at each node of the BFS tree induced (Lemma~\ref{lemma:bfsrootedbfs}); hence dispersion is achieved.

In one iteration of the main {\bf while} loop:
\begin{enumerate}
\item
{\tt DetermineLeafDemand(i)} does 2-neighborhood traversals in parallel, and hence takes $O(\Delta^2)$ rounds.
\item
In each of the $\Delta$ iterations of the {\bf while} loop of {\tt PopulateNextLevel(i)},
the upward movement and the downward movement of the robots in lines (12) and (17-18), respectively, takes $i$ rounds; and the code block (20-25) takes $2\Delta$ rounds. 
\end{enumerate}
So the time complexity is $\Delta^2 + \Delta(2\Delta+2i)$. By Lemma~\ref{lemma:bfsrootedbfs}, a BFS tree is induced and hence the maximum number of levels is $D$, which is the number of iterations of the {\bf while} loop of {\em Graph\_Disperse\_BFS}. 
Thus the overall time time complexity is $\sum_{i=1}^D 1(\Delta^2 + \Delta(2\Delta+2i))$ = $O(D\Delta(\Delta + D))$.

The variable $nrobots$ takes $\log\,k$ bits, $level$ and $i$ take $\log\,D$ bits each, $demand[1\ldots\Delta]$ takes $\Delta\,log\,k$ bits, $child\_id[1\ldots\Delta]$ takes $\Delta\,log\,k$ bits, $parent\_id$ takes $\log \,k$ bits, $parent$ takes $\log\Delta$ bits, and $winner$ takes 1 bit. 

Note that the local communication model suffices because the broadcast function can be simulated by each settled robot moving up to its parent and back, to communicate the demand of its subtree.  The theorem follows.
\end{proof}

\begin{algorithm*}[!h]
{\footnotesize
Initialize: $nrobots\gets$ number of robots; $level$, $i$, $demand[1\ldots\Delta]$, $child\_id[1\ldots\Delta]$; $parent\_id$, $parent$, $winner$\\
robot with lowest ID settles at root, $level\gets 0$\\
\While{$nrobots>0$}{
DetermineLeafDemand($i$)\\
PopulateNextlevel($i$)\\
$i\gets i+1$
}
\underline{DetermineLeafDemand($i$)}\\
Each settled robot $r$ at a leaf node $u$ at level $i$ does a 2-bounded DFS to count number of neighbors $v$ at level $i+1$. If on exploring $(u,v)$ via $out_u$, (i) $v$ is level $i-1$, then backtrack, (ii) else if $v$ has a level $i-1$ neighbor, then $v$ is level $i$ node - discount and backtrack, (iii) else $v$ is a level $i+1$ node, hence robot $r$ sets $demand_u[out_u]\gets 1$.\\
Wait until $\Delta^2$ rounds are elapsed.\\
\underline{PopulateNextLevel($i$)}\\
\While{$[\sum_{leaf \, u}\sum_{j=1}^{\Delta} (1 \mbox{ if } demand_u[j]=1)] > 0 \bigwedge [\sum_{leaf \, u}\sum_{j=1}^{\Delta} (1 \mbox{ if } demand_u[j]=2)] < nrobots$}{
All unassigned robots at level $i$ move upwards to root using $parent$ pointers in $i$ rounds\\
Leaf node $u$ broadcasts $B1(my\_id,parent\_id,\sum_{j=1}^{\Delta} (1 \, if \, demand_u[j]=1))$\\
On receiving $B1(x,my\_id,y)$, if $child\_id^{-1}[x]=\theta$ then $demand[\theta]\gets y$\\
On receiving $B1$ from all children ($\forall x \,|\,child\_id[x]\neq 0$), broadcast  $B1(my\_id,parent\_id,\sum_{j=1}^{\Delta} demand[j])$\\
Wait until $i$ rounds are elapsed; synchronize()\\
When root receives $B1$ from all children, distribute $\min(nrobots,\sum_{j=1}^{\Delta} demand[j])$ robots among children after resetting $winner\gets 0$ for each robot\\
Robots move down the tree to the leaf nodes at level $i+1$: On receiving $x$ robots, a node at level $<i$ ($=i$) distributes among children reachable via ports $p$ such that $demand[p]>0$ ($demand[p]=1$). \\
On arrival at level $i+1$ node $v$ from level $i$ node $u$ via $(out_u,in_v)$, robot with lowest ID sets $winner\gets 1$\\
\If{$winner=0$}{
retrace back via $in_v$ to $u$ and wait
}
\ElseIf{$winner=1$}{
visit each neighbor $w$ via $(out_v,in_w)$. If $w (\neq u)$ is at level $i$, $demand_w[in_w]\gets 3$\\
retrace back from $v$ using $in_v$ to $u$; $demand_u[out_u]\gets 2$
}
Wait until $2\Delta-1$ rounds are elapsed since arriving at level $i+1$ (so all robots at level $i+1$ are back at level $i$); synchronize()
}
All unassigned robots at level $i$ move upwards to root using $parent$ pointers in $i$ rounds\\
Leaf node $u$ broadcasts $B1(my\_id,parent\_id,\sum_{j=1}^{\Delta} (1 \, if \, demand_u[j]=2))$\\
On receiving $B1(x,my\_id,y)$, if $child\_id^{-1}[x]=\theta$ then $demand[\theta]\gets y$\\
On receiving $B1$ from all children ($\forall x \,|\,child\_id[x]\neq 0$), broadcast  $B1(my\_id,parent\_id,\sum_{j=1}^{\Delta} demand[j])$\\
Wait until $i$ rounds are elapsed.\\
When root receives $B1$ from all children, distribute $\min(nrobots,\sum_{j=1}^{\Delta} demand[j])$ robots among children; $nrobots\gets nrobots - \sum_{j=1}^{\Delta} demand[j]$\\
Robots move down the tree to the leaf nodes at level $i$: On receiving $x$ robots, a node at level $<i$ distributes among children reachable via ports $p$ such that $demand[p]>0$. \\
At a level $i$ node $u$, on receiving $\leq \sum_{j=1}^{\Delta} (1 \, if \,demand_u[j]=2)$ robots, send one robot $x$ on each port $out \,|\, demand_u[out]=2$; $child\_id_u[out]\gets x$\\
The robot $x$ reaches node $v$ at level $i+1$ via incoming port $in$, $level\gets i+1$, $parent\_id\gets u$, $parent\gets in$, $x$ settles at the node
}
\caption{Algorithm {\em Graph\_Disperse\_BFS} to solve \dis\ in global model. $r$ denotes a settled robot, if any, at that node.}
\label{algo:algo2bfs}
\end{algorithm*}
\setlength{\textfloatsep}{0pt}

\subsection{General Case}
We adapt the single-rooted algorithm to the multi-rooted case. From each root, a BFS tree is initiated in parallel, and is identified by the robot ID settled at the root. When two (or more) BFS trees meet at a node, a collision is detected; the tree with the higher depth (if unequal depths) subsumes the other tree(s) and collects the robots of the other tree(s) at its root. It then continues the BFS algorithm at the same depth in case  $nrobots =0$, i.e., level $i+1$ may not be fully populated yet.  A collision of two trees $T_x$ and $T_y$ is identified by a 4-tuple for each tree: $\langle root,depth,bordernode,borderport\rangle$. The changes to Algorithm~\ref{algo:algo2bfs} are given next, and the module for Collision processing is given in \onlyLong{Algorithm~\ref{algo:graphcollision}.} \onlyShort{the full version attached in Appendix.}

\begin{enumerate}
    \item After Line 5, insert a line: Invoke a call to Collision processing\onlyLong{~(Algorithm~\ref{algo:graphcollision})}.
    \item Line 6: Conditionally increment $i$ in line 6, as described in Step 3 of Collision processing\onlyLong{~(Algorithm~\ref{algo:graphcollision})}.
    \item Line 8: The case (iii) becomes case (iv) and the new case (iii) is: if $v$ has a settled robot $r'$ of another tree with root $root'$ and level $lvl'$, $r$ broadcasts Collide($\langle root,i+1,u,out_u\rangle,\langle root',lvl',v,in_v\rangle$) -- then discount $v$ and backtrack.
    \item Lines 19-24: are to be executed only with reference to robots belonging to my own tree (having same root).
    \item Line 34: execute if no other robot from any tree arrives at the node $v$. Otherwise execute line 35.
    \item Add new Line 35 in Algorithm~\ref{algo:algo2bfs}: For each other robot $r'$ of tree with root $root'$ and level $i+1$, broadcast 
    Collide($\langle root,i+1,u,out_u\rangle,\langle root',i+1,r',v_{in}\rangle$). Wait for SUBSUME messages broadcast in Collision processing\onlyLong{~(Algorithm~\ref{algo:graphcollision})}. If $x$'s tree subsumes other trees, $x$ sets $level\gets i+1$, $parent\_id\gets u$, $parent\gets in$, $x$ settles at $v$. (If $x$'s tree is subsumed, $x$ retraces its step back to $u$, then moves on to the root of the tree that subsumes its tree as described in Collision processing.)
\end{enumerate}

\onlyLong{
\begin{algorithm*}[!t]
{\normalsize
The root node $r_{min}$ with the lowest ID among those roots of trees that are involved in collisions does Collision processing. Using all Collide messages broadcast in this iteration of the {\bf while} loop of line 3, Algorithm~\ref{algo:algo2bfs}, $r_{min}$ creates a undirected Collison graph $G_C=(V_C,E_C)$. 
$V_C=\{T.root\,|\, Collide(T,*) \mbox{ or } Collide(*,T) \mbox{ message is received }\}$.
$E_C=\{(T.root,,T'.root)\,|\, \mbox{ at least one } Collide(T,T') \mbox{ message is received } \}$.
$r_{min}$ creates a Maximal Independent Set (MIS) from among those nodes $T.root$ of $V_C$ having $T.depth= i+1$. It then creates a partition $P=\{P_1,\ldots P_{|MIS|}\}$ of $V_C$ such that each $P_a$ has exactly one node, $central(P_a)$, of the MIS and a subset of $\overline{MIS}$ nodes that $central(P_a)$ covers. Such a partition is feasible because each Collide message has at least one parameter $T_x$ such that 
$T_x.depth=i+1$. $r_{min}$ then broadcasts 
one SUBSUME($T_x,T_y$) record corresponding to some received $Collide(T_x,T_y)$ message, for each $T_x.root$ such that  $T_x.root = central(P_a).root$ and each $T_y.root$ such that $T_y.root$ is covered by $T_x.root$ in $P_a$. In effect, this  selects, for each pair of BFS trees that collide and are assigned to the same partition $P_a$, one pair (among possibly multiple pairs) of border nodes in the two trees via which one tree (corresponding to $central(P_a)$) will subsume the other tree.  \\
For each partition $P_a$ of $G_C$, all robots in the trees of $G$ corresponding to $\overline{MIS}\cap P_a$ nodes of $G_C$ collect to $central(P_a).root$ root node of that tree of $G$  corresponding to the MIS member node $central(P_a)$ in its partition of $G_C$.
For record SUBSUME$(T_x,T_y)$, $T_y.bordernode$ identifies a path $H$ from $T_y.root$ to $T_y.bordernode$ by nodes along $H$ serially broadcasting B2($my\_id,parent\_id,parent$) progressively up the path from $T_y.bordernode$ to $T_y.root$ ($i+1$ serial broadcasts suffice).
All robots in $T_y$, except those robots along $H$, beginning from the leaf node robots move up tree $T_y$ to $T_y.root$ using the $parent$ pointers ($i+1$ rounds suffice). The robots in $T_y$ then move down path $H$ from $T_y.root$ to $T_y.bordernode$ ($i+1$ rounds suffice); then to $T_x.bordernode$ and then up to $T_x.root$ using the $parent$ pointers ($i+1$ rounds suffice). \\
$nrobots$ at $T_x.root$ is then updated with the count of the newly arrived robots, and another iteration of $Graph\_Disperse\_BFS$ is executed in the tree corresponding to $T_x$. This iteration is for the same value of $i$ if $nrobots =0$ before the update, i.e., $i$ is not incremented in line 6 of Algorithm~\ref{algo:algo2bfs} immediately following the Collision processing because level $i+1$ may not be fully populated as yet and to maintain the BFS property, level $i+1$ needs to be filled before filling level $i+2$.

}
\caption{Module for Collision processing for multi-rooted case of {\em Graph\_Disperse\_BFS} to solve \dis\ in global model. $r$ denotes a settled robot, if any, at that node.}
\label{algo:graphcollision}
\end{algorithm*}
}

\begin{theorem}
Algorithm~\ref{algo:algo2bfs} ({\em Graph\_Disperse\_BFS}) along with the modifications given in this section solves \dis\ in multi-rooted graphs in 
$O(\max(D,k)\Delta(\Delta + D))$ rounds and requires $O(\max(D,\Delta\log\,k))$ memory in the global communication model.
\label{th:bfsmultiroottime}
\end{theorem}
\onlyLong{
\begin{proof}
Each concurrently-initiated BFS tree grows until it collides with another (or runs out of robots). When two (or more) trees collide and are assigned to the same partition of $G_C$, one tree subsumes the robots of the other(s) and by continuing the same logic for the same value of level $i$ (step 3 of Algorithm~\ref{algo:graphcollision}), the BFS-tree property is preserved until termination (exhausting all the unsettled robots at the root of the tree). On termination, one robot is settled at each distinct node of the BFS tree(s). Hence \dis\ is achieved.

Building on the proof of Theorem~\ref{th:bfsrootedtime}, in addition to $O(\sum_{i=1}^D \Delta^2 + i\Delta)$ $=O(\Delta D (\Delta +D))$ rounds, we have more iterations of the main {\bf while} loop. Let there be $x$ such iterations corresponding to $x$ serial subsumptions by the tree in question. Let the depth (level) of the tree for the $j$th subsumption be $d_j$. Then the following additional time cost is incurred:  $\sum_{j=1}^x \Delta^2$ (for lines 22-24 of {\tt PopulateNextLevel}) $+ 3d_j$ (for robot movements in step 2 of Algorithm~\ref{algo:graphcollision}) $+ 2d_j\Delta$ (for robot movements in {\tt PopulateNextLevel}). $d_j<D$ and $x<k$. Thus the additional time $\leq \Delta^2 k + 3Dk + 2kD\Delta$. The time complexity follows.

The logic introduced in the multi-rooted algorithm does not add variables that increase the bit complexity. 

The global communication model needs to be used for issuing and processing the broadcasts. The theorem follows.
\end{proof}
}

{\bf Proof of Theorem~\ref{theorem:1}:} Follows from Theorems~\ref{th:bfsrootedtime} and \ref{th:bfsmultiroottime}.

\section{Algorithm for Arbitrary Trees (Theorem \ref{theorem:2})}
\label{section:algorithm2}
We present and analyze a BFS-based algorithm that solves {\dis} of $k\leq n$ robots on an arbitrary $n$-node tree in $O(D\max(D,k))$ time with $O(\max(D,\Delta\log\,k))$ bits of memory at each robot in the global communication model. 
This algorithm improves the $O(n)$ time of the best previously known algorithm \cite{Augustine:2018} for arbitrary trees (Table \ref{table:comparision-tree}) in the local communication model.

In the rooted graph case, we first show that dispersion takes $O(D^2)$ time with $O(\max(D,\Delta\log\,k))$ bits of memory at each robot in the local communication model. We then extend the result to the general case of robots on multiple nodes. 

\onlyShort{\vspace{-3mm}}
\subsection{Rooted Case}
\label{rootedtree}
The rooted tree case, where all robots are initially colocated at one node, is a special case of Algorithm~\ref{algo:algo2bfs} adapted to the tree topology. {\tt DetermineLeafDemand(i)} behaves as follows: instead of lines 8-9, $demand_u[out_u]\gets 1$ for all $out_u$ other than $parent$. In {\tt PopulateNextLevel(i)}, the {\bf while} loop (lines 11-25) is removed, and robots do not move upwards to root (line 26). A single iteration, as per the following modification of lines (27-34), and excluding line (30), suffices.

\begin{description}
\item[27:] Leaf node $u$ broadcasts 
$B1(my\_id,parent\_id,\delta-1)$.
\item[28:] On receiving $B1(x,my\_id,y)$, if $child\_id^{-1}[x]=\theta$ then $demand[\theta]\gets y$.
\item[29:] On receiving $B1$ from all children ($\forall x \,|\,child\_id[x]\neq 0$), broadcast  $B1(my\_id,parent\_id,\-\sum_{j=1}^{\Delta} demand[j])$.
\item[31:] When root receives $B1$ from all children, distribute $\min(nrobots,\sum_{j=1}^{\Delta} demand[j])$ robots among children; 
$nrobots\gets nrobots - \sum_{j=1}^{\Delta} demand[j]$.
\item[32:] Robots move down the tree to the leaf nodes at level $i$: On receiving $x$ robots, a node at level $<i$ distributes among children reachable via ports $p$ such that $demand[p]>0$. 
\item[33:] At a level $i$ node $u$, on receiving $\leq \delta-1$  
robots, send one robot $x$ on each port $out \,|\, demand_u[out]=1$; $child\_id_u[out]\gets x$.
\item[34:] The robot $x$ reaches node $v$ at level $i+1$ via incoming port $in$, $level\gets i+1$, $parent\_id\gets u$, $parent\gets  in$, $x$ settles at the node.
\end{description}

\begin{theorem}
Algorithm~\ref{algo:algo2bfs} ({\em Graph\_Disperse\_BFS})  with the changes described in this section solves 
\dis\ on a single-rooted tree in $O(D^2)$ rounds and requires $O(\max(D,\Delta\log\,k))$ memory in the local communication model.
\label{th:bfsrootedtreetime}
\end{theorem}
\onlyLong{
\begin{proof}
As this logic is a special case of Algorithm~\ref{algo:algo2bfs}, a BFS tree of settled robots is created (Theorem~\ref{th:bfsrootedtime}) and \dis\ is solved. 

Based on the modifications above to Algorithm~\ref{algo:algo2bfs}, the main {\bf while} loop runs $i=1$ to $D$, whereas inside, the traversal of the robots down the tree from the root to the leaf nodes (line 32) takes $i$ rounds in iteration $i$. So the number of rounds is $\sum_{i=1}^D i = O(D^2)$.
The bit space complexity is same as for the rooted graph case, and the local communication model is also used as explained in the proof of Theorem~\ref{th:bfsrootedtime}.
The theorem follows.
\end{proof}
}

\onlyShort{\vspace{-2mm}}
\subsection{General Case}
\label{multirootedtree}
Each of the multiple roots on the tree topology initiate in parallel the execution of the rooted tree case, as described in Section~\ref{rootedtree}, with the following additional changes.
\begin{enumerate}
    \item After Line 5 of Algorithm~\ref{algo:algo2bfs}, insert a line: Invoke a call to Collision processing\onlyLong{~(Algorithm~\ref{algo:graphcollision})}.
    \item Line 6 of Algorithm~\ref{algo:algo2bfs}: Conditionally increment $i$ in line 6, as described in Step 3 of Collision processing\onlyLong{~(Algorithm~\ref{algo:graphcollision})}.
    \item Execute line 34 of Algorithm~\ref{algo:algo2bfs} only if no other robot from any other tree  arrives at the node $v$ nor is there a robot from another tree settled at $v$: ``The robot $x$ reaches node $v$ at level $i+1$ via incoming port $in$, $level\gets i+1$, $parent\_id\gets u$, $parent\gets  in$, $x$ settles at the node.'' Otherwise execute the following line 35.
    \item Add new line 35 in Algorithm~\ref{algo:algo2bfs}: For each other robot $r'$ of tree with root $root'$ and level $i+1$ that arrives at $v$, broadcast 
    Collide($\langle root,i+1,u,out_u\rangle,\langle root',i+1,r',v_{in}\rangle$). 
    If there is a robot $r''$ of tree with root $root''$ and level $i'' \leq i$ settled at $v$, broadcast
    Collide($\langle root,i+1,u,out_u\rangle,\langle root'',i'',r'',v_{in}\rangle$).
    Wait for SUBSUME messages broadcast in Collision processing\onlyLong{~(Algorithm~\ref{algo:graphcollision})}. 
    If $x$'s tree subsumes other trees and there is no robot $r''$ settled at $v$, $x$ sets $level\gets i+1$, $parent\_id\gets u$, $parent\gets in$, $x$ settles at $v$. 
    Else if $x$'s tree subsumes other trees and there is a robot $r''$ settled at $v$, $x$ retraces back to $u$ and up to its root, along with robots of other subsumed trees that relocate to $root$ (the root of the tree associated with $x$), and also resets $child\_id_u[out]$.
    (Else if $x$'s tree is subsumed, $x$ retraces its step back to $u$, then moves on to the root of the tree that subsumes its tree as described in Collision processing.)
\end{enumerate}

\begin{theorem}
Algorithm~\ref{algo:algo2bfs} ({\em Graph\_Disperse\_BFS}) along with changes described above and in Section~\ref{rootedtree} solves \dis\ on a multi-rooted tree in $O(D\max(D,k))$ rounds and requires $O(\max(D,\Delta\log\,k))$ memory in the global communication model.
\label{th:bfsmultiroottreetime}
\end{theorem}
\onlyLong{
\begin{proof}
Each concurrently-initiated BFS tree grows until it collides with another (or runs out of robots). When two (or more) trees collide and are assigned to the same partition of $G_C$, one tree subsumes the robots of the other(s) and by continuing the same logic for the same value of level $i$ (step 3 of Algorithm~\ref{algo:graphcollision}), the BFS-tree property is preserved until termination (exhausting all the unsettled robots at the root of the tree). On termination, one robot is settled at each distinct node of the BFS tree(s). Hence \dis\ is achieved.

Building on the proofs of Theorems~\ref{th:bfsrootedtime},~\ref{th:bfsmultiroottime}, and ~\ref{th:bfsrootedtreetime}, in addition to $\sum_{i=1}^D i$ $=O(D^2)$ rounds, we have to account for steps of the robot movements when a tree subsumes others. Let there be $x$ such serial subsumptions by the tree in question. Let the depth (level) of the tree for the $j$th subsumption be $d_j$. Then the following additional time cost is incurred:  $\sum_{j=1}^x 3d_j$ (for robot movements in step 2 of Algorithm~\ref{algo:graphcollision}) $+ d_j$ (for robot movements in line 32 of {\tt PopulateNextLevel}). $d_j<D$ and $x<k$. Thus the additional time $\leq 4Dk$. The time complexity follows.
The logic introduced in the multi-rooted algorithm does not add variables that increase the bit complexity. 
The global communication model needs to be used for issuing and processing the broadcasts. The theorem follows.
%
\end{proof}
}

{\bf Proof of Theorem~\ref{theorem:2}:} Follows from Theorems~\ref{th:bfsrootedtreetime} and \ref{th:bfsmultiroottreetime}.

\onlyShort{\vspace{-2mm}}
\section{Concluding Remarks}
\label{section:conclusion}

We have presented three results for solving {\dis} of $k\leq n$ robots on $n$-node graphs. The first two results are for arbitrary graphs and the third result is for arbitrary trees.
The first result for arbitrary graphs is based on a DFS traversal and improves by $O(\log k)$ factor the best previously known algorithm in the local communication model. 
The second algorithm for arbitrary graphs is based on a BFS traversal and improves significantly on the $O(\Delta^D)$ time of the best previously known algorithm in the local communication model. 
The algorithm for arbitrary trees is also based on a BFS traversal and improves on the $O(n)$ time of best previously known algorithm in the local communication model.


For future work, it will be interesting to solve {\dis} on arbitrary graphs using a DFS-based algorithm with time $O(k)$ or improve the existing time lower bound of $\Omega(k)$ to $\Omega(\min(m,k\Delta))$. 
For BFS-based algorithms, it will be interesting to improve $\max(D,k)$ factor to $O(D)$ for both arbitrary graphs and trees.   The third interesting direction will be to consider faulty robots; our algorithms as well as previous algorithms \cite{Augustine:2018,Kshemkalyani,Kshemkalyani2019} assume fault-free robots. The fourth interesting direction will be to solve {\dis} in dynamic graphs; so far only static graph cases are studied. The fifth interesting direction will be to extend our algorithms to solve {\dis} in semi-synchronous and asynchronous settings.  

\balance
\bibliographystyle{plain}
\bibliography{references}

\begin{thebibliography}{10}

\bibitem{Augustine:2018}
John Augustine and William~K. Moses, Jr.
\newblock Dispersion of mobile robots: {A} study of memory-time trade-offs.
\newblock {\em CoRR}, abs/1707.05629, [v4] 2018 (a preliminary version appeared
  in ICDCN'18).

\bibitem{Bampas:2009}
Evangelos Bampas, Leszek G{a}sieniec, Nicolas Hanusse, David Ilcinkas, Ralf
  Klasing, and Adrian Kosowski.
\newblock Euler tour lock-in problem in the rotor-router model: I choose
  pointers and you choose port numbers.
\newblock In {\em DISC}, pages 423--435, 2009.

\bibitem{Barriere2009}
L.~Barriere, P.~Flocchini, E.~Mesa-Barrameda, and N.~Santoro.
\newblock Uniform scattering of autonomous mobile robots in a grid.
\newblock In {\em IPDPS}, pages 1--8, 2009.

\bibitem{Cohen:2008}
Reuven Cohen, Pierre Fraigniaud, David Ilcinkas, Amos Korman, and David Peleg.
\newblock Label-guided graph exploration by a finite automaton.
\newblock {\em ACM Trans. Algorithms}, 4(4):42:1--42:18, August 2008.

\bibitem{Cormen:2009}
Thomas~H. Cormen, Charles~E. Leiserson, Ronald~L. Rivest, and Clifford Stein.
\newblock {\em Introduction to Algorithms, Third Edition}.
\newblock The MIT Press, 3rd edition, 2009.

\bibitem{Cybenko:1989}
G.~Cybenko.
\newblock Dynamic load balancing for distributed memory multiprocessors.
\newblock {\em J. Parallel Distrib. Comput.}, 7(2):279--301, October 1989.

\bibitem{Das18}
Shantanu Das, Dariusz Dereniowski, and Christina Karousatou.
\newblock Collaborative exploration of trees by energy-constrained mobile
  robots.
\newblock {\em Theory Comput. Syst.}, 62(5):1223--1240, 2018.

\bibitem{Das16}
Shantanu Das, Paola Flocchini, Giuseppe Prencipe, Nicola Santoro, and Masafumi
  Yamashita.
\newblock Autonomous mobile robots with lights.
\newblock {\em Theor. Comput. Sci.}, 609:171--184, 2016.

\bibitem{Dereniowski:2015}
Dariusz Dereniowski, Yann Disser, Adrian Kosowski, Dominik Pajak, and
  Przemyslaw Uzna\'{n}ski.
\newblock Fast collaborative graph exploration.
\newblock {\em Inf. Comput.}, 243(C):37--49, August 2015.

\bibitem{DisserMNSS16}
Yann Disser, Frank Mousset, Andreas Noever, Nemanja Skoric, and Angelika
  Steger.
\newblock A general lower bound for collaborative tree exploration.
\newblock {\em CoRR}, abs/1610.01753, 2016.

\bibitem{ElorB11}
Yotam Elor and Alfred~M. Bruckstein.
\newblock Uniform multi-agent deployment on a ring.
\newblock {\em Theor. Comput. Sci.}, 412(8-10):783--795, 2011.

\bibitem{Flocchini2012}
Paola Flocchini, Giuseppe Prencipe, and Nicola Santoro.
\newblock {\em Distributed Computing by Oblivious Mobile Robots}.
\newblock Synthesis Lectures on Distributed Computing Theory. Morgan {\&}
  Claypool Publishers, 2012.

\bibitem{flocchini2019}
Paola Flocchini, Giuseppe Prencipe, and Nicola Santoro.
\newblock {\em Distributed Computing by Mobile Entities}, volume~1 of {\em
  Theoretical Computer Science and General Issues}.
\newblock Springer International Publishing, 2019.

\bibitem{FraigniaudGKP06}
Pierre Fraigniaud, Leszek Gasieniec, Dariusz~R. Kowalski, and Andrzej Pelc.
\newblock Collective tree exploration.
\newblock {\em Networks}, 48(3):166--177, 2006.

\bibitem{Fraigniaud:2005}
Pierre Fraigniaud, David Ilcinkas, Guy Peer, Andrzej Pelc, and David Peleg.
\newblock Graph exploration by a finite automaton.
\newblock {\em Theor. Comput. Sci.}, 345(2-3):331--344, November 2005.

\bibitem{Hsiang:2003}
Tien-Ruey Hsiang, Esther~M. Arkin, Michael~A. Bender, Sandor Fekete, and Joseph
  S.~B. Mitchell.
\newblock Online dispersion algorithms for swarms of robots.
\newblock In {\em SoCG}, pages 382--383, 2003.

\bibitem{HsiangABFM02}
Tien{-}Ruey Hsiang, Esther~M. Arkin, Michael~A. Bender, S{\'{a}}ndor~P. Fekete,
  and Joseph S.~B. Mitchell.
\newblock Algorithms for rapidly dispersing robot swarms in unknown
  environments.
\newblock In {\em WAFR}, pages 77--94, 2002.

\bibitem{Kshemkalyani}
Ajay~D. Kshemkalyani and Faizan Ali.
\newblock Efficient dispersion of mobile robots on graphs.
\newblock In {\em ICDCN}, pages 218--227, 2019.

\bibitem{Kshemkalyani2019}
Ajay~D. Kshemkalyani, Anisur~Rahaman Molla, and Gokarna Sharma.
\newblock Improved dispersion of mobile robots on arbitrary graphs.
\newblock {\em CoRR}, abs/1812.05352, 2018.

\bibitem{MencPU17}
Artur Menc, Dominik Pajak, and Przemyslaw Uznanski.
\newblock Time and space optimality of rotor-router graph exploration.
\newblock {\em Inf. Process. Lett.}, 127:17--20, 2017.

\bibitem{tamc19}
Anisur~Rahaman Molla and William K.~Moses Jr.
\newblock Dispersion of mobile robots: The power of randomness.
\newblock In {\em TAMC}, pages 481--500, 2019.

\bibitem{Ortolf:2012}
Christian Ortolf and Christian Schindelhauer.
\newblock Online multi-robot exploration of grid graphs with rectangular
  obstacles.
\newblock In {\em SPAA}, pages 27--36, 2012.

\bibitem{Poudel18}
Pavan Poudel and Gokarna Sharma.
\newblock Time-optimal uniform scattering in a grid.
\newblock In {\em ICDCN}, pages 228--237, 2019.

\bibitem{Shibata:2016}
Masahiro Shibata, Toshiya Mega, Fukuhito Ooshita, Hirotsugu Kakugawa, and
  Toshimitsu Masuzawa.
\newblock Uniform deployment of mobile agents in asynchronous rings.
\newblock In {\em PODC}, pages 415--424, 2016.

\bibitem{Subramanian:1994}
Raghu Subramanian and Isaac~D. Scherson.
\newblock An analysis of diffusive load-balancing.
\newblock In {\em SPAA}, pages 220--225, 1994.

\end{thebibliography}


\end{document}